\documentclass[11pt,letterpaper]{article}
\usepackage[letterpaper,margin=1in]{geometry}
\usepackage{amsmath,amsthm,amssymb}
\usepackage{hyperref}
\usepackage[T1]{fontenc}
\newtheorem{proposition}{Proposition}

\newtheorem{definition}{Definition}[section]

\newtheorem{theorem}[definition]{Theorem}

\newtheorem{corollary}[definition]{Corollary}

\newtheorem{lemma}[definition]{Lemma}

\newcommand{\vcg}{\mathrm{VCG}}
\newcommand{\ervcg}{\text{$\mathrm{rVCG}$} }
\newcommand{\te}{\mathrm{TE}}
\newcommand{\MSW}{\mathrm{MSW}}
\newcommand{\bb}{\mathbf b}
\newcommand{\bp}{\mathbf p}
\newcommand{\bx}{\mathbf x}
\newcommand{\mech}{{\cal M}}
\newcommand{\bv}{\mathbf v}
\newcommand{\bt}{\mathbf t}

\newcommand{\bxt}{\mathbf{ \tilde x}}

\begin{document}
\title{Approaching Utopia: Strong Truthfulness and Externality-Resistant Mechanisms}

\author{Amos Fiat\thanks{Tel Aviv University. fiat@tau.ac.il} \and Anna R. Karlin\thanks{University of
Washington. karlin@cs.washington.edu} \\ \and Elias
Koutsoupias\thanks{University of Athens. elias@di.uoa.gr}
\and Angelina Vidali\thanks{University of Vienna. angvid@gmail.com}}

\maketitle

\begin{minipage}{\textwidth}\scriptsize
 ``And verily it is naturally given to all men to esteem their own inventions
best.''
\par
--- Sir Thomas More, in \emph{Utopia}, Book 1, 1516 AD.
%{\sl De optimo reip. statv, deque noua insula Vtopia, libellus uere aureus,
%nec minus salutaris quam festiuus} (Utopia), Book 1, 1516
%AD.
\end{minipage}

%\vskip 1cm

\begin{abstract}
We introduce and study strongly truthful mechanisms and their
applications. We use strongly truthful mechanisms as a tool for
implementation in undominated strategies for several problems,
including the design of externality resistant auctions and a variant of
multi-dimensional scheduling.
\end{abstract}

\section{Introduction}

\subsection{Externalities}
Mechanisms with externalities, and specifically altruism and spite,
but also others (``the joy of winning'', ``malice''),  have been
studied at length in the literature. Experiments seem to indicate
that both altruism and spite  have an observable effect, and
various theoretical models have been proposed to deal with this
issue.

 We quote higher authority (Cooper and Fang \cite{Cooper}) in the context of 2nd price auctions:\newline
{\sl``We found that small and medium overbids are more likely to
occur when bidders perceive their rivals to have similar values,
supporting a modified `joy of winning' hypothesis but large
overbids are more likely to occur when bidders believe their
opponents to have much higher values, consistent with the `spite'
hypothesis.''}

A partial list of (experimental and theoretical) references dealing
with externalities is
\cite{Kagel,Ledyard,Levine,Jehiel,Brandt,Morgan,Moscibroda,Babaioff,AaronRoth,Brandt,Maasland,Cooper,ChenKempeBaysian,ChenKempe,poanchen,DBLP:conf/wine/ChenKKS11}.
The questions addressed in previous work primarily deal with the
impact of externalities on the equilibria, e.g., observing that
externalities such as ``the joy or winning'' or ``spite'' lead to
overbidding in some auction mechanisms, or that externalities
modeled as altruism lead to more-or-less balanced outcomes in the
ultimatum game, although neither of these phenomena would be
considered ``reasonable'' if one assumes no externalities. In recent
years, the price of anarchy as impacted by such externalities has
also been the subject of much research, {\sl e.g.}, malice in congestion games
\cite{Moscibroda,Babaioff,AaronRoth}.

In this paper we consider a somewhat different goal:  we seek to
devise mechanisms that {\sl overcome} externalities. As a basic
motivating example consider an auction selling a single item. The
Vickrey second-price auction is dominant strategy incentive compatible.
But, try to imagine that the bidders who lose are spiteful towards the
winner (although this is really hard to believe). They may have reason to
increase their bid so as to increase the payment by the winning
agent.

Even more worrisome --- say that the only spiteful losers are those
who
 took part in the various experimental psychology studies cited above,
 and they did so only so as to mislead the researchers.
 In fact, we who reside in Utopia will never, ever, encounter spite.
 This is a fact, but it does {\sl not} imply that
everyone {\sl believes} that it is so, it does {\sl not} imply that
everyone believes that everyone believes that it is so, it does {\sl
not} imply that this is common knowledge. Ergo, just the {\sl
concept} of spite (transmitted via the apple from the Garden of
Eden), even if in fact there {\sl are no} spiteful bidders, implies
that bidders may have an incentive not to bid truthfully in the  VCG
mechanism.

So, why not define the agent type to include all possible
externalities and then run VCG? There are two problems here: (1)
It is impossible; payments to one agent impact the utility of
another, we are no longer in the quasi-linear setting, (2) Ignoring
the former concern ({\sl i.e.}, impossibility), what social welfare
are we optimizing? Is it our goal to pander to the spiteful masses?
Offer them bread and circuses? Execute the winners during the lunch
break of the Gladiatorial games? There are indications from the
lives of the Caesars that this may actually maximize (spiteful)
social welfare.

So, the very existence of the concept of spite seems to threaten the
fundamentals of mechanism design.

To address these issues, 
we study an alternative utility model: We assume that 
agents have two utility functions, a {\em base
utility}, and an {\em externality-modified utility} which is a linear
combination of other agent utilities. Variants of this model appear
in Ledlard \cite{Ledyard}, Levine \cite{Levine}, Chen and Kempe
\cite{ChenKempeBaysian,ChenKempe,poanchen}, and many other papers.
The PhD thesis of Chen \cite{poanchen} includes numerous relevant
papers.

\subsection{Externality Resistant Mechanisms}

We present a new type of private value mechanism, $\ervcg$. Assume
that it is common knowledge that no one is willing to lose more than (say)
$\gamma=5\mathrm{\ cents}$ so as to increase another's payment by
$\$1$. Now: \begin{enumerate}
\item
 Agents using the $\ervcg$
mechanism are sure that the following two values are
approximately equal: \begin{itemize} \item The utility they obtain
under $\ervcg$, in an imperfect world, where externalities are real,
and demons roam the earth. \item The utility they would have
obtained under $\vcg$, in an imaginary, Utopian world, where
externalities did not exist. (See Theorem \ref{thm:mainexttheorem}).
 \end{itemize}
 {\sl I.e.}, given a bound, $\gamma$, on the altruism/spite, the $\ervcg$ mechanism approximates Utopia, as promised
in the title.\footnote{Admittedly, the bound $\gamma$ has to be very small in order
to truly approximate Utopia.}
\item On the other hand, irrespective of how infinitesimally small $\gamma>0$ may be, a
losing bidder in a second-price auction, may, out of spite, even
infintesimally small spite, reduce the winner's profit to zero.
(This holds in the more general $\vcg$ mechanism as well).
\end{enumerate}

\subsection{Strongly Truthful Mechanisms}

To achieve externality resistant mechanisms we make use of strongly
truthful mechanisms. These are mechanisms where it is not only a
weakly dominant strategy to be truthful but where one gets punished
for lying. The goal in the design of strongly truthful mechanisms is
to increase the punishment as much as possible. Strongly truthful
mechanisms are related to strongly convex mechanisms, analogous to
the connection between truthful mechanisms and convex utility
functions, (see, {\sl e.g.}, Archer and Kleinberg,
\cite{DBLP:conf/sigecom/ArcherK08,DBLP:journals/sigecom/ArcherK08}).

For bounded domains, we give (optimal)  strongly truthful
mechanisms, in this case, the punishment for the lie $\tilde{v} = v
+ \delta$ is $O(\delta^2)$.

 For unbounded domains, we give a mechanism
that is {\sl relatively strongly truthful} where the lie is measured
as a fraction of the truth, and the punishment for the lie
$\tilde{v}=(1+\alpha)v$, where $\alpha\in \Theta(1)$, is
$v/\log^{1+\epsilon} v$.

Strongly truthful mechanisms can also be used in mechanisms for
multi-dimensional problems such as makespan minimization for
unrelated machines, see below.

This idea of combining multiple mechanisms to boost truthfulness
appears in \cite{Nissim:2012:AOM:2090236.2090254}, where it is used
to derive truthful mechanisms for some problems via differential
privacy. It also appears implicitly in the context of scoring rules
\cite{BRIER,Bickel01062007}, and in the related responsive lotteries
\cite{DBLP:conf/sagt/FeigeT10a} so as to determine the true utility
of an outcome. However, we are unaware of previous attempts to
quantify the quality of such devices, nor are we aware of
other attempts to apply them towards externality resistance or for multidimensional
problems.

In the appendix we describe transformations between strongly
truthful mechanisms and proper scoring rules. This automatically
implies transformations between strongly truthful mechanisms, market
scoring rules, responsive lotteries, and market maker pricing
algorithms to provide liquidity for prediction markets
\cite{Hanson:2007:1750-676X:3,NBERw10504,Chen07autility}.

\subsection{The Solution Concept}

Adapting a solution concept from  Babaioff, Lavi and Pavlov,
\cite{DBLP:conf/soda/BabaioffLP06}, from approximation problems to
arbitrary
 predicates, we say that a  mechanism $M$ is an algorithmic implementation of a predicate
$P$ in undominated strategies, if, for all agents $i$, there exists
a set of strategies, $D_i$, such that
\begin{enumerate}
\item The output of $M$ satisfies $P$, for any combination of
strategies from $\prod_j D_j$, and, \item For all $i$, for any agent
$i$ strategy, $s\notin D_i$, there exists some strategy $s'\in D_i$
that is strictly better for agent $i$ than strategy $s$,
irrespective of what strategies are chosen by the other agents.
%Moreover, such a strategy can be computed in polynomial
%time.
\end{enumerate}

{\sl I.e.},  predicate $P$ is implemented by a mechanism in
undominated strategies, if, in the game defined by the mechanism,
and as long as no agent chooses a strategy that is obviously
dominated (for arbitrary assumptions about the types of other
agents, {\sl e.g.}, values, bids, externalities), predicate $P$
holds for the outcome of the mechanism.

In the context of externality resistant auctions, as long as agents
do not bid stupidly (do not use a strategy that is obviously
dominated), externality resistance holds.

In fact, any strategy that entails bidding ``too far away'' from the
truth is dominated  by bidding truthfully, where ``too far away'' for
agent $i$ is a function of  her own externalities
$\gamma_{ij}$ (see Section \ref{sec:extern} for a
definition of these externalities). Moreover, agents can efficiently determine
that bidding far from the truth is dominated by truthful bidding.
Thus,  $D_i$ is contained in the set of all
bids whose distance from the truth is not too big. Note that we
don't make any claim on the precise strategies that will be adapted
by the agents.

%This solution concept seems quite reasonable, much more so than
%other solution concepts such as iterative removal of dominated
%strategies. Or Nash equilibria in undominated strategies. See
%discussion in \cite{Jackson}.

\subsection{Other Applications}

We can also use strongly truthful mechanisms to achieve goals such
as minimizing the makespan in a multi dimensional machine scheduling
problem, the infamous Nisan-Ronen problem, see
\cite{NR99,LS07,ADL09}.

This magic is achieved by changing the problem, and allowing one to
repeatedly assign the same job to a machine. So, choosing to verify
the truthfullness of the agent types can be done by choosing at
random, with some small probability, a target agent, and using
strongly truthful mechanisms to punish the agent for
misrepresentation of his type.

Given a sufficiently large punishment, all agents will have
incentive to stick close to the truth. So, with high probability,
the mechanism will achieve a close approximation to the minimum
makespan in undominated strategies..

This is quite general and can be used in other multi-dimensional
settings where one can boost truth extraction by repetition.

%%% Local Variables: 
%%% mode: latex
%%% TeX-master: "stronglytruthful"
%%% End: 

\section{Strongly Truthful Mechanisms}

A key ingredient in our constructions is the notion of a {\em strongly truthful mechanism}. In
this section, we define strongly truthful mechanisms for single
dimensional problems and one agent. As discussed below, these
definitions and results extend to multi-dimensional and
multi-agent settings.

Consider a single dimensional agent with private value (type) $v$
for receiving a good or service. A direct revelation mechanism takes
as input some (possibly false) value, $\tilde{v}$, computes a
payment, $p(\tilde{v})$, and allocates the good to the agent with
probability $a(\tilde{v})$.

The standard quasilinear utility of an agent whose true value is $v$,
but reports value $\tilde{v}$ (possibly different from $v$), is
denoted by \begin{equation} u_v(\tilde{v}) = v \cdot a(\tilde{v}) -
p(\tilde{v}).\label{eq:utility} \end{equation} We also define
$$u(v)=u_v(v),$$ {\sl i.e.}, the utility to the agent with value $v$
when truthfully reporting $\tilde{v}=v$.

In this setting, it follows from Myerson \cite{My81}, that a
mechanism is truthful in expectation if and only if
\begin{itemize}
\item the allocation probability function $a(v)$ is monotone
nondecreasing, and
\item the payment function is $$p(v) = v a(v) - \int_0^v a(x) dx +
p(0),$$ for some constant $p(0)$. (We will take $p(0)=0$ herein).
\end{itemize}

It follows from the above and from Equation \ref{eq:utility} that
\begin{eqnarray*} u_v(\tilde{v}) &=& v \cdot a(\tilde{v}) - \tilde{v} \cdot
a(\tilde{v}) + \int_0^{\tilde{v}} a(x) dx \\
&=& (v-\tilde{v}) \cdot a(\tilde{v}) + \int_0^{\tilde{v}} a(x) dx.
\end{eqnarray*}

Thus, for truthful in expectation mechanisms, it must be that
\begin{enumerate}
  \item The utility function  $u(v)$ is convex (the integral of
  a nondecreasing function).
  \item The allocation function $a(v) =  u'(v)$ (where $u$ is differentiable).
  \item Any convex function $u(v)$ whose subgradient, $ u'(v)$, lies in the range $[0,1]$, can be
  interpreted as the utility function for an associated truthful in expectation mechanism.
  \item Ergo, if restricting oneself to truthful in expectation mechanisms,
  one can describe a mechanism using utility functions or allocation
  functions interchangeably. (Up to additive constants).
\end{enumerate}

We seek to strengthen the notion of truthfulness in expectation so
that the greater the deviation from the truth, the greater the loss
in utility.

To this end, we define $c$-strongly truthful mechanisms
as follows:
\begin{definition}
A mechanism with utility function $u$ is called
$c$-strongly truthful if for every $v$ and $\tilde{v}$:
\begin{equation}
u_v(v)-u_v(\tilde{v})\geq \frac{1}{2} c\, |\tilde{v} -
v|^2. \label{eq:strongtruth}
\end{equation}
\end{definition}
This definition enables us  to extend the connection between
truthfulness and convexity to strongly truthful mechanisms. For
this, recall the standard notion of strong convexity. For a
differentiable function $f(x)$, convexity is equivalent to:
$$\forall x, x' \quad f(x)-f(x')\geq f'(x')\cdot (x-x').$$
The following notion is also standard~\cite{boyd2004convex}:
\begin{definition}
Let $m\ge 0$. A function $f$ is called
$m$-strongly convex if and only if for every $x$, $x'$:
\begin{equation}
\label{eq:5} f(x)-f(x')\geq  f'(x')\cdot (x-x') + \frac{1}{2}
m \, |x-x'|^2
\end{equation}
\end{definition}
By defining strong truthfulness as in Equation
\eqref{eq:strongtruth} the following proposition holds:
\begin{lemma}
A mechanism with utility function $u(v)$ is $m$-strongly truthful
if and only if $u(v)$ is $m$-strongly convex.
\end{lemma}
\begin{proof}
Applying equation \eqref{eq:utility} to
$u_v(\tilde{v})$ and $u_{\tilde{v}}(\tilde{v})$, we get
\[
u_{\tilde{v}}(\tilde{v})-u_v(\tilde{v})= u'(\tilde{v})\cdot
(\tilde{v} - v) .\]
It follows that
\begin{equation}
\label{eq:4} u_v(v)-u_v(\tilde{v})=u_v(v)-u_{\tilde{v}}(\tilde{v})+
u'(\tilde{v})\cdot(\tilde{v}-v)
=u(v)-u(\tilde{v})+
u'(\tilde{v})\cdot(\tilde{v}-v)
\end{equation}
Since by definition, the mechanism is $m$-strongly truthful if
and only if $u_v(v)-u_v(\tilde{v})\geq \frac{1}{2} m\,
|\tilde{v} - v|^2$, we derive that the mechanism is $m$-strongly truthful if
and only if $u(v)-u(\tilde{v})+
u'(\tilde{v})\cdot(\tilde{v}-v)\geq \frac{1}{2} m\,
|\tilde{v} - v|^2$, which is precisely the definition that $u(v)$
is $m$-strongly convex.
\end{proof}

\noindent{\em Remark:}
All of the definitions in this section extend naturally to
multi-dimensional agents. Indeed, the three equivalent definitions of
a doubly differentiable function being convex (the standard one, cycle
monotonicity, and the Hessian being positive semidefinite) have
analogues when discussing truthful multidimensional mechanisms over
convex domains~\cite{boyd2004convex,rockafellar1997convex}. Similarly,
the equivalent notions of strong-convexity and strong truthfulness
extend mutatis mutandis.

It follows from the above theorem that the question of finding the
strongest truthful mechanism is an extremal question about strongly
convex functions whose partial derivatives satisfy appropriate
constraints that capture the constraints of the allocation
probabilities (for example, for the single item case
the constraint is the derivative of the utility is in $[0,1]$).

\subsection{Strongly Truthful Mechanisms for Single Agent, Single Item Auctions}

Consider the case in which we want to find the strongest truthful mechanism
for a single player and one item. (We will use this in the next section.)
To start, assume that the agent's value is bounded: $v\in [L,H]$.
For this case, we define the {\em linear mechanism}:

\begin{definition}
\label{defn:lin-mech}
The linear mechanism for the single player/single item setting has
allocation rule $a(v)= (v-L)/(H-L)$, and applies when the player's value
is known to be in the range $[L,H]$.
\end{definition}

\begin{theorem}
\label{thm:lin-mech}
The linear mechanism for a player whose value $v$ satisfies $v\in [L,H]$
is a $1/(H-L)$-strongly truthful mechanism.
No other mechanism is $m$-strongly truthful for $m\geq 1/(H-L)$.
\end{theorem}
\begin{proof}
It is straightforward to check that for the linear mechanism, the
utility function $u(v)$ is  $\frac{(v-L)^2}{2(H-L)}$.
We can directly verify that Equation \eqref{eq:5}
in the definition of strong convexity holds with equality for all $v$, with
$m=1/(H-L)$. Indeed, we derive the following equivalences
\begin{align*}
u(z)-u(y)&=  u'(y)\cdot (z-y) + \frac{1}{2} m \, |z-y|^2 \\
\frac{(z-L)^2}{2(H-L)}-\frac{(y-L)^2}{2(H-L)}&=
\frac{y-L}{H-L}(z-y)+\frac{1}{2} \frac{1}{H-L} (z-y)^2
\\
\frac{(z-y)(z+y-2L)}{2(H-L)}&=  \frac{(z-y)(2(y-L)+(z-y))}{2(H-L)};
\end{align*}
the last equality clearly holds.

We now show that this is the strongest truthful mechanism. From the
definition of strong convexity for the extreme values of the domain,
we get
\begin{align*}
u(H)-u(L)&\geq u'(L) (H-L) + \frac{1}{2} m (H-L)^2 \\
u(L)-u(H)&\geq u'(H) (L-H) + \frac{1}{2} m (L-H)^2
\end{align*}
Adding these two, we get that
\[
(H-L)(u'(H)-u'(L)) \geq m (H-L)^2
\]
Since $u'(L)$, and $u'(H)$ are in $[0,1]$ (they represent allocations), we get that $m\leq
1/(H-L)$.
\end{proof}

\noindent{\em Remarks:}
\begin{itemize}
\item There is a direct connection between single-agent truthful
mechanisms and scoring rules (see e.g., \cite{Bickel01062007}). 
Indeed, one can define a notion of {\em strongly
proper scoring rules} that is analogous to a strongly truthful mechanism.
We note that the mechanism just described is in fact
the well-known quadratic scoring rule.

\item Definition \ref{defn:lin-mech} and Theorem \ref{thm:lin-mech}
can easily be generalized to the case of one player and many items
with additive valuations. In this case the utility is
$u(v)=\sum_{j=1}^m \frac{(v_j-L)^2}{2(H-L)}$, for which
$m=\frac{1}{n(H-L)}$.

\end{itemize}

\subsection{Relative strong truthfulness}

If we want to consider unbounded domains, it
follows from Theorem \ref{thm:lin-mech} that no $m$-strongly
truthful mechanism exists with $m>0$.

For such domains, it may be useful to define a notion of relative strong truthfulness
as follows:
\begin{definition}
We say a mechanism $M$ is $f(v, \alpha)$-relatively truthful if, for
all $\tilde{v}$ such that $\tilde{v}\not\in [v(1 - \alpha), v(1+\alpha)]$
$$\frac{u_v(v) - u_v(\tilde{v})}{u_v(v)} \ge f(v, \alpha).$$
\end{definition}
For example, it is easy to show that the single agent mechanism with
allocation rule $a(v)= 1- \frac{1}{\ln v}+\frac{1}{\ln^2 v}$ (and
payment $p(v)=\frac{v}{\ln^2(v)}$) satisfies $f(v, \alpha) =
\Omega(\frac{\alpha^2}{\log^2 v})$. Slightly better mechanisms that
approach $f(v, \alpha) = \Omega(\frac{\alpha^2}{\log v})$
exist\footnote{The following sequence of mechanisms are defined for
  every $k>1$ and approach $f(v, \alpha) = \Omega(\frac{\alpha^2}{\log v})$:
\begin{align*}
p(v)&=\frac{v}{\ln^k v} & a(v)&=1-\frac{1}{(k-1) \ln^{k-1}
  v}+\frac{1}{\ln^k v}  %& u(v)&=v-\frac{v}{(k-1) \ln^{k-1} v} 
\\
p(v)&=\frac{v}{\ln v\ln^k \ln v} & a(v)&=1-\frac{1}{(k-1) \ln v\ln^{k-1}
  \ln v}+\frac{1}{\ln v\ln^k \ln v} %& u(v)&=v-\frac{v}{(k-1) \ln v\ln^{k-1} \ln v} 
\end{align*}
and so on.}. 

%%% Local Variables: 
%%% mode: latex
%%% TeX-master: "stronglytruthful"
%%% End: 

\section{Externality Resistant Auctions}
\label{sec:extern}

In this section, we consider how strongly truthful, or
truth-extraction mechanisms can be used to help cope with spiteful or
altruistic bidders. Our goal is to ensure that a bidder participating
in, say, an auction for a single item, does not need to worry about
her competitor purposely bidding high just so as to make her pay a
lot.

We consider the setting where an auctioneer wishes to maximize
social welfare, and each agent has a value $v_i$ for being one of
the winners in the auction. Of course, in the standard version of
this setting, the mechanism of choice would be the VCG mechanism.

As we have already discussed however, the VCG mechanism is entirely
vulnerable to spiteful agents. Before explaining the alternative we
propose, we define a  utility model for externalities that captures
precisely what we mean when we speak about spiteful and altruistic
agents.

In the {\em externality-modified} setting, agent i's type $t_i$
consists of
\begin{itemize}
\item $v_i$, her value for service; and
\item a set of externality parameters $\gamma_{ij}$ for all
$j\ne i$. Intuitively, $\gamma_{ij}$ represents how much agent $i$
cares about the utility of agent $j$. A large, negative value means
that $i$ is significantly motivated by the desire to decrease agent
$j$'s utility, whereas a large, positive value means that $i$ seeks
to increase agent $j$'s utility. A value of zero means that $i$ is
indifferent towards $j$.
\end{itemize}

Let $\mech$ be an arbitrary mechanism for the single-parameter
allocation problem under consideration. The mechanism takes as input
a bid $b_i$ from each agent (which is equal to $v_i$ if the mechanism is truthful)
and produces as output an allocation
$\bx$, where $x_i$ is the probability that agent $i$ receives
service, and payments $\bp$, with $p_i$ the expected payment by
agent $i$. Note that both $x_i= x_i(\bb)$ and $p_i= p_i(\bb)$ are
functions of the bids. The allocation selected must satisfy the feasibility
constraints of the setting, however, we do assume, that having only a single
arbitrary agent receive service is feasible.

Given bids $\bb_{-i}$ of all players except player $i$, the {\em base
(standard) utility} of agent $i$, when her type is
$t_i=(v_i,\{\gamma_{i1}, \ldots, \gamma_{in}\})$ and her bid is $b_i$,
is denoted by $u_{v_i} ^{\mech}( b_i, \bb_{-i})$ and is defined as
\begin{equation}u_{v_i} ^{\mech}( b_i, \bb_{-i}) = v_i \, x_i(\bb) - p_i
(\bb).\label{eq:ui}\end{equation}

Notice that this utility depends only on $v_i$ and not on the rest of
agent $i$'s private information (agent $i$'s
externality parameters $\gamma_{ij}$). This is why we subscript the utility
by $v_i$ instead of $t_i$.

We define the {\em externality-modified} utility
$\widehat{u}_{t_i}^{\mech}$ of agent $i$ when the types of the
agents are $\bt$ and the bids are $\bb$ as
\begin{equation}\widehat{u}_{t_i} ^{\mech}(\bb,\bt_{-i}) = u_{v_i} ^{\mech}(
b_i, \bb_{-i})
 + \sum_{j\ne i} \gamma _{ij} \, u_{v_j}^{\mech}( b_j,
 \bb_{-j}).\label{eq:uhati}\end{equation}
This model (and variants thereof) have been used previously in
several papers, e.g. \cite{Levine,poanchen}.

Note that \begin{itemize} \item Because the externality-modified
utility defined above depends, not only on the bids (or actions) of other
agents, $b_{-i}$, but also on their types, we add the $t_{-i}$ as an
argument to the utility function, which is atypical. (Of course
the only part of $t_{-i}$ the utility depends on is $v_{-i}$.)
\item The value of $t_{-i}$ is, in general, unknown to
agent $i$, so agent $i$ will be, in general, unable to compute her
externality modified utility
$\widehat{u}_{t_i}^{\mech}(\bb,\bt_{-i})$. 
\item 
We will be particularly interested in cases
where the mechanism $\mech$ that is being run is $\vcg$, and
then use 
$u_{v_i}^{\vcg}(\bb)$ to denote the standard utility of agent $i$
when her value is $v_i$, the reported bids are $\bb$ and the mechanism
being run is $\vcg$.
\end{itemize}

%We consider private value mechanisms,
%that are {\sl not} dominant
%strategy incentive compatible, the setting is as follows:
%\begin{itemize}
%\item We assume that agent $\ell$  {\sl knows} her own $v_\ell$,
%$b_\ell$, and the values $\gamma_{\ell,j}$, she also knows that
%$|u_i| \leq 1$ and that  $|\gamma_{ij}| \leq \gamma$ for all
%$i,j$.
%\item Other than the bounds above, agent $\ell$
%does {\sl not} necessarily know $v_i$, $b_i$, or $\gamma_{ij}$ for
%any $i\neq \ell$ and $j$.
%\item Moreover, constrained by the bounds above, agent $\ell$ may have
%an arbitrary ``world view" about the values of $v_i$, $b_i$, and
%$\gamma_{ij}$ for $i\neq \ell$ and $j$. We denote such an imaginary
%value by using the $\tilde{}$\ notation, {\sl i.e.}, $\tilde{v}_i$,
%$\tilde{b}_i$, $\tilde{\gamma}_{ij}$, $i\neq \ell$, are imaginary
%values, bids, and externalities, and $\tilde{\bv}_{-\ell}$,
%$\tilde{\bb}_{-\ell}$, are vectors of imaginary values and
%externalities.
%\item We prove our externality-resistent guarantees,
%Theorem \ref{thm:mainexttheorem} below, in the private value
%setting, by arguing that irrespective of the agent's world views,
%the claims of the theorem hold.
%\end{itemize}

Our goal is to design a mechanism that is {\em
externality-resistant}, in the following sense:
\begin{itemize}
  \item The mechanism approximately maximizes social welfare.
  \item Despite the fact that each agent bids to maximize their externality-modified utility, each agent ends up with
  a base utility that is approximately what it would have been had all agents bid so as to maximize their base utility. Thus,
  non-spiteful agents are not harmed by the presence of spiteful
  agents. Furthermore, the auctioneer's revenue is not harmed by the
  presence of altruistic agents.
\end{itemize}

An immediate difficulty that arises is the fact that our utility
model is non-quasi linear. Our approach is to consider a weaker
solution concept, {\em implementation in undominated strategies,}
formally defined as follows.

In a game of incomplete information, a strategy for an agent is a
function mapping types to actions. We say that a strategy $s_i'$ for
agent $i$ is {\em dominated} by strategy $s_i$ if for all types
$t_i$ for agent $i$, and for all possible types $t_{-i}$ and all
possible actions $$b_{-i}=s_{-i}(t_{-i})$$ of the other agents, the
utility of agent $i$ satisfies:
$$ u_{t_i}(s_i(t_i), b_{-i}, t_{-i}) \ge u_{t_i}(s'_i(t_i), b_{-i}, t_{-i}).$$
A strategy $s_i$ for agent $i$ is {\em undominated} if it is not
dominated.

We say that a mechanism implements a predicate $P$ in undominated
strategies, if whenever agents are limited to playing undominated
strategies, it must be that predicate $P$ holds.

We consider the following simple variant of VCG, which we call {\em
externality-resistant} VCG or \ervcg, for short. The \ervcg mechanism,
with $n$ agents participating, is parameterized by a value $\delta$,
$0 \le \delta \le 1$, and works as follows:

\begin{itemize}
\item Ask the $n$ agents for their values/bids.
\item With probability $\delta/n$ single out the $i$-th agent and run the
truth extraction mechanism (denoted by TE) on him.
\item With probability $1-\delta$, run VCG.
\end{itemize}

Our main theorem is the following:
\begin{theorem} \label{thm:mainexttheorem}

Consider any single-parameter allocation setting in which the $n$ agents
true values $v_i$ are all in the range $[0,1]$. Let $\gamma$ denote
$\max_{ij} \gamma_{ij}$.  Then, for any $n$, $\delta$, $\epsilon$, and
$\gamma_{ij}$ such that $$\gamma< {\frac
{\epsilon \, \delta}{ 8\,\left( 1-\delta \right)^2 {n}^{3}}},$$ mechanism \ervcg above
implements the following predicate in undominated strategies:

For all agents $i$, and for all types $\bt$, the
base utility obtained under the \ervcg mechanism is close to
the base utility obtained by agent $i$ 
when all agents bid truthfully
under the ``standard'' VCG mechanism. Specifically, for $\bb$
undominated,
$$ u_{v_i}^{\ervcg}(\bb)\ge (1- \delta) u_{v_i}^{\vcg}(\bv)
- \epsilon.$$

%such that $\gamma_{ij}=0$ for all $j\neq i$, {\sl
%i.e.}, for all agents that have no externalities, and for all types
%$\bt$, the utility obtained under the \ervcg mechanism is close to
%the utility obtained by agent $i$ when all agents bid truthfully
%under the ``standard'' VCG mechanism: 
%\begin{eqnarray*} (\gamma_{ij}
%= 0 \quad \forall j \neq i)\qquad \Rightarrow \qquad
%u_{t_i}^{\ervcg}(\bb, \bt_{-i})&\ge& (1- \delta) u_{t_i}^{\vcg}(\bv)
%- \epsilon.
%\end{eqnarray*}
\end{theorem}

 In the proof below, we use the following notation and definitions:

 \begin{itemize} \item For any set of
 bids $\bb$, let
$\MSW(\bb)$ denote the maximum social welfare achievable with
respect to the bids $\bb$, i.e.
$$\MSW (\bb) = \max_a  \sum_{j} b_j(a).$$

\item We define $\MSW_{v_i}(b_i, \bb_{-i})$ to be the maximum social
welfare {\em experienced} by agent $i$, when agent $i$ bids $b_i$,
whereas her true value is $v_i$, and all other agents bid
$\bb_{-i}$. Thus,
$$\MSW _{v_i} (b_i, \bb_{-i}) = v_i (a^*) + \sum_{j\ne i} b_j(a^*),
\mbox{\rm \ where \ }a^* = \mathrm{argmax}_a \sum_{j} b_j(a).$$

%\item If agent $i$ bids $b_i$, her true value is $v_i$, and her world view is that
%the other agents bid $\tilde{\bb}_{-i}$. We say that
%$\MSW_{v_i}(b_i, \tilde{\bb}_{-i}) = \MSW_{v_i}(\tilde{\bb}_i)$ is
%the maximal social welfare {\em imagined} by agent $i$.

\item When agent $i$ bids
$b_i$, her true value is $v_i$, all agents but $i$ bid $\bb_{-i}$,
then, the utility of agent $i$ under VCG with Clarke Pivot Payments
is
$$u^\vcg_{v_i}(\bb)=\MSW _{v_i} (b_i, \bb_{-i}) - \MSW ( \bb_{-i}). $$

%\item When agent $i$ bids
%$b_i$, her true value is $v_i$, and her world view is that all
%agents but $i$ bid $\tilde{\bb}_{-i}$, then, her {\sl imaginary
%utility} under VCG with Clarke Pivot Payments is $u^\vcg_i(v_i, b_i,
%\tilde{\bb}_{-i}) = u^\vcg(v_i,\tilde{\bb}).$ We use the shorthand
%notation $u^\vcg_i(v_i, \tilde{\bb}_{-i})$ to represent
% $u^\vcg_i(v_i, v_i, \tilde{\bb}_{-i}))$.

\end{itemize}

We can now turn to the proof of the theorem.

\begin{proof}

We assume that agents would like to maximize their
externality-modified utility:
\begin{equation} \widehat{u}_{t_i}^{\ervcg}(\bb,\bt_{-i}) = (1- \delta)
\widehat{u}_{t_i}^{\vcg}(\bb,\bt_{-i}) + \frac{\delta}{n}
\left(u_{v_i}^{\te}(b_i) + \sum_{j\ne i} \gamma_{ij}
u_{v_j}^{\te}(b_j)\right),\label{eqn:exutil}\end{equation} where
\begin{equation}
\widehat{u}_{t_i}^{\vcg}(\bb,\bt_{-i}) = u_{v_i}^{\vcg}(\bb) + \sum_{j\ne
i}\gamma_{ij} \left(u_{v_j}^{\vcg}(\bb)\right). \label{eqn:vcg-ext}
\end{equation}
%and some players seek to maximize their standard utility:
%\begin{equation}u_{t_i}^{\ervcg}(\bb,\bt_{-i}) = (1- \delta) u_{t_i}^{\vcg}(\bb) +
%\frac{\delta}{n} u_{t_i}^{\te}(b_i).\label{eqn:stdutil}
%\end{equation}

We say an agent is {\em standard} if
 $\gamma_{ij}=0$ for all $j$. Such an agent doesn't care
about the utility of others. For standard agents, the base utility
and the externality-modified utility are the same.
In addition, 
because VCG is dominant strategy truthful, 
 it is a
dominant strategy for the standard agents to bid truthfully.

So, we only need to understand how non-standard agents will
bid. To this end, fix an agent $i$ of type $t_i$ and the bids $\bb_{-i}$ and
types $\bt_{-i}$ of the other agents.

% Consider any such  agent $i$, let $\bb_{-i}$ and $\bv_{-i}$ be
%the bids and values of the other agents, and let $\tilde{\bb}_{-i}$
%and $\tilde{\bv}_{-i}$ be the world view of agent $i$ regarding the
%bids and values of the other agents. Then, while agent $i$ may {\sl
%think} that she is maximizing the utility given in Equation
%\ref{eqn:exutil}, she actually tries to maximize her {\sl imaginary
%utility}:
%\begin{equation} \widehat{u}_i^{\ervcg}(v_i, b_i, \tilde{\bb}_{-i}) = (1- \delta)
%\widehat{u}_i^{\vcg}(v_i, b_i, \tilde{\bb}_{-i}) + \frac{\delta}{n}
%\left(u_i^{\te}(v_i, b_i) + \sum_{j\ne i} \gamma_{ij}
%u_j^{\te}(\tilde{v}_j,
%\tilde{b}_j)\right),\label{eqn:exutilimag}\end{equation} where
%\begin{equation}
%\widehat{u}_i^{\vcg}(v_i, b_i, \tilde{\bb}_{-i}) = u_i^{\vcg}(v_i,
%b_i, \tilde{\bb}_{-i}) + \sum_{j\ne i}\gamma_{ij}
%\left(u_j^{\vcg}(\tilde{v}_j, \tilde{b}_j, b_i,
%\tilde{\bb}_{-i,j})\right), \label{eqn:vcg-extimag}
%\end{equation}.

Suppose that the VCG part of \ervcg is executed. Then agent $i$'s
externality-modified utility $\widehat{u}_{t_i}^{\vcg}(b_i,
\bb_{-i}, \bt_{-i})$ is defined by equation (\ref{eqn:vcg-ext}). We
have
$$u^\vcg_{v_j}(\bb)=\MSW_{v_j} ({b}_j, {\bb}_{-j}) - \MSW ({\bb}_{-j}) $$
and
$$u^\vcg_{v_j}(b_j, v_i,{\bb}_{-ij})=\MSW_{{v}_j} ({b}_j, v_i,{\bb}_{-ij}) - \MSW (v_i, {\bb}_{-i,j}). $$
Thus, the difference between agent $i$'s externality modified
utility when agent $i$ bids $b_i$ and when agent $i$ bids $v_i$ is
given by
\begin{eqnarray} \widehat{u}_{t_i}^{\vcg}(b_i,
{\bb}_{-i}, \bt_{-i} ) - \widehat{u}_{t_i}^{\vcg}( v_i, {\bb}_{-i}, \bt_{-i} )
& = &  \MSW_{v_i}(b_i, {\bb}_{-i}) - \MSW(v_i, {\bb}_{-i})\nonumber\\
& &  +
 \sum_{j\ne i}\gamma_{ij}
\left(\MSW_{{v}_j}({b}_j, {\bb}_{-j}) - \MSW_{{v}_j}({b}_j, v_i, {\bb}_{-ij})\right)\nonumber \\
& &  +  \sum_{j\ne i}\gamma_{ij}\left( \MSW(v_i, {\bb}_{-i,j}) - \MSW (b_i,{\bb}_{-i,j})\right) \nonumber\\
&\le & 2  \sum_{j\ne i}\gamma_{ij} \eta_i  \nonumber\\
&\le & 2(n-1) \gamma_i \eta_i,
\label{eqn:vcg-ext-bound}
\end{eqnarray}

where $|b_i - v_i|= \eta_i$ and $\gamma_i = \max_{j} \gamma_{ij}$.

On the other hand,
\begin{eqnarray}
\widehat{u}_{v_i}^{\te}(v_i) - \widehat{u}_{v_i}^{\te}(b_i ) \ge
\frac{1}{2} m(b_i) |b_i - v_i|^2 \ge \frac{1}{2} \eta_i^2,
\label{eqn:te-ext}
\end{eqnarray}
assuming valuations in the range $[0,1]$ and the use of the linear $\te$ algorithm.
Combining inequalities (\ref{eqn:vcg-ext-bound}) and (\ref{eqn:te-ext}), we have
$$\widehat{u}_{t_i}^{\ervcg}(b_i, {\bb}_{-i} ) - \widehat{u}_{t_i}^{\ervcg}(v_i, {\bb}_{-i} )
\le (1-\delta) 2(n-1) \gamma \eta_i -
 \frac{\delta}{2n} \eta_i^2.$$
Consider any bid $b_i$ for which  $\eta_i= |b_i-v_i|$ has
$$(1-\delta) 2(n-1) \gamma \eta_i -
 \frac{\delta}{2n} \eta_i^2 < 0.$$
Then the strategy of bidding truthfully dominates the strategy of bidding this $b_i$.
Thus, for all undominated strategies, it must be that agent $i$ bids a value $b_i$
which satisfies:
$$0 \le \widehat{u}_{t_i}^{\ervcg}(b_i, {\bb}_{-i} , \bt_{-i})
- \widehat{u}_{t_i}^{\ervcg}(v_i, {\bb}_{-i}, \bt_{-i} ),$$ implying
that  she will choose $\eta_i$ so that
$$0 \le (1-\delta) 2(n-1) \gamma \eta_i -
 \frac{\delta}{2n} \eta_i^2,$$
and thus
%$$(1-\delta)2(n-1) \gamma \eta_i \ge  \frac{\delta}{2n} \eta_i^2,$$
\begin{equation}\eta_i \le \frac{4(1-\delta)}{\delta} n^2\gamma.
\label{eqn:noLieCond}
\end{equation}

In other words, \emph{lying about $v_i$ by more than the right-hand side
of Equation (\ref{eqn:noLieCond}) is a strategy dominated by the
truth-telling strategy}.

%For a standard player $\ell$ with $\gamma_{\ell,j}=0$ for
%all $\ell\ne j$ who participates in $\ervcg$, we can guarantee that
%if all agents play undominated strategies, then:
%\begin{eqnarray}
%u_{t_{\ell}}^{\ervcg}(v_{\ell}, \bb_{-\ell},\bt_{-\ell})&  =& (1-
%\delta) u_{v_{\ell}}^{\vcg}(\bb_{-\ell})
%+ \frac{\delta}{n} u_{v_{\ell}}^{\te}(v_\ell) \nonumber\\
%& = & (1- \delta)\left( u_{t_{\ell}}^{\vcg}(v_\ell, \bv_{-\ell})  -
%\left(u_{v_{\ell}}^{\vcg}(v_\ell, \bv_{-\ell}) -
%u_{v_{\ell}}^{\vcg}(v_\ell, \bb_{-\ell})\right)\right)  \nonumber \\
%& & + \frac{\delta}{n} u_{v_{\ell}}^{\te}(v_\ell) \nonumber\\
%& \ge & (1- \delta) \left( u_{v_{\ell}}^{\vcg}(v_\ell, \bv_{-\ell}) -2 n \eta\right) \nonumber\\
%& \ge & (1- \delta) u_{v_{\ell}}^{\vcg}(v_\ell, \bv_{-\ell}) -
%\frac{8(1-\delta)^2}{\delta} n^3\gamma
%\end{eqnarray}
%where $\eta = \max_i \eta_i$.

From this we can conclude that for any player $\ell$ who participates in $\ervcg$, 
if all agents play undominated strategies, then:
\begin{eqnarray}
u_{v_{\ell}}^{\ervcg}(b_{\ell}, \bb_{-\ell})&  =& (1-
\delta) u_{v_{\ell}}^{\vcg}(b_{\ell},\bb_{-\ell})
+ \frac{\delta}{n} u_{v_{\ell}}^{\te}(b_\ell) \nonumber\\
& = & (1- \delta)\left( u_{v_{\ell}}^{\vcg}(v_\ell, \bv_{-\ell})  -
\left(u_{v_{\ell}}^{\vcg}(v_\ell, \bv_{-\ell}) -
u_{v_{\ell}}^{\vcg}(b_\ell, \bb_{-\ell})\right)\right)  \nonumber \\
& & + \frac{\delta}{n} u_{v_{\ell}}^{\te}(b_\ell)\nonumber\\
& \ge & (1- \delta) \left( u_{v_{\ell}}^{\vcg}(v_\ell, \bv_{-\ell}) -2 n \eta\right) \nonumber\\
& \ge & (1- \delta) u_{v_{\ell}}^{\vcg}(v_\ell, \bv_{-\ell}) -
\frac{8(1-\delta)^2}{\delta} n^3\gamma
\end{eqnarray}
where $\eta = \max_i \eta_i$.

%Similarly, for a spiteful player, participating in $\ervcg$ with some number of
%other spiteful players, we will have
% \begin{eqnarray}
%\widehat{u}_i^{\ervcg}(v_i, b_i, \bb_{-i}) &= &
%(1- \delta) \widehat{u}_i^{\vcg}(v_i, b_i, \bb_{-i})
%+ \frac{\delta}{n} \left(u_i^{\te}(v_i, b_i) +
%\sum_{j\ne i} \gamma_{ij} u_j^{\te}(v_j, b_j)\right)\\
%&\ge& \ldots \mathrm{something}.
%\end{eqnarray}
\end{proof}
The following two corollaries are immediate:
\begin{corollary}
When $\ervcg$ is used and  all players play undominated strategies,
the social welfare of the outcome $a^*$ selected satisfies
$$\sum_iv_i(a^*) \ge \MSW(\bv) - n\eta.$$ 
\end{corollary}
\begin{corollary}
When $\ervcg$ is used and  all players play undominated strategies,
the profit of the auctioneer is at least his profit from running 
VCG with truthful players minus $2n\eta$.
\end{corollary}

%%% Local Variables: 
%%% mode: latex
%%% TeX-master: "stronglytruthful"
%%% End: 

\section{Discussion}

In this paper, we have introduced a number of concepts and taken
first steps towards understanding and applying these concepts. Clearly
though, we have only scratched the surface.

For example, the basic tool of strongly truthful mechanisms may have
some potential, but there is clearly much left to be understood.
What is the right way to define strong truthfulness? What mechanisms
achieve optimal relative strong truthfulness? What is the tradeoff
between the ``strength'' of the truthfulness and the social welfare
that can be achieved?

We explored a new utility model for externalities in the context of
mechanism design and sought to design mechanisms that protect agents
from these externalities. Our mechanism has some
externality-resistance, however,  the externality parameters have to
be extremely small in order for our mechanism to be effective. Is it
possible to do better? More concretely, our mechanism tolerates
externality parameters of value $\gamma=O(1/n^3)$. Are there
mechanisms that tolerate higher values of $\gamma$? In the opposite
direction, can we show a bound on the maximum $\gamma$?

What happens if we use a different solution concept? Also, while we
chose to optimize for ``base utility'', that is not the only goal one
might consider. One could optimize for social welfare with respect to
the externality modified utilities.  To what extent is this possible?
Is this a reasonable goal? {\sl I.e.}, should the goal of the
mechanism be to encourage spite?

\bibliographystyle{plain}
\bibliography{strongbib}

\appendix
\section{Strong Truthfulness and Scoring Rules}

\newcommand{\f}{\tilde}
\newcommand{\R}{\mathbb{R}}

\newcommand{\bpt}{\mathbf{ \tilde p}}

In this appendix, we develop the connection between strongly
truthful single-agent mechanisms and ``strongly proper'' scoring
rules \cite{BRIER,Bickel01062007}. As mentioned above, this
immediately relates strongly truthful mechanisms to a host of
seemingly unrelated problems.

\subsection{The setting}

We consider the following setting:

\begin{itemize}
\item There is a set of $n+1$ events, (call them events 0 through $n$)
one of which will happen.
\item The agent (forecaster) has a belief vector $\bp$ as to which event
will happen, where $p_i$ is the probability that event $i$ happens.
$p_0=1 - \sum_{1\le i \le n} p_i$.
\item The mechanism takes as input a postulated belief vector $\bpt$, and
uses a scoring rule to determine the ``payments'' or ``scores''.
Specifically, the scoring rule says for each outcome $i$,  the ``payment''
or ``score'' the agent gets is $s_i(\bpt)$.
\item The agent proposes $\bpt$ and obtains utility
$$\sum_{0\le i\le n} p_i s_i(\bpt).$$
\item The scoring rule is strictly proper if reporting $\bpt=\bp$
strictly  maximizes his utility. 
\end{itemize}

\subsection{Translating Mechanisms to Scoring Rules}

Let $M$ be a mechanism that takes as input an agent's valuations
$x_1, \ldots, x_n$ for $n$ alternatives.
We assume $x_i \ge 0$ for all $i$ and that $\sum_i x_i \le 1$.
The mechanism has allocation probabilities $a_i(\bx)$ and a payment
rule $P(\bx)$. 

We convert this to a scoring rule $S(M)$  as follows: Given vector
$\bp$ representing the probabilities $p_1, \ldots, p_n$ of outcomes
(with $p_0= 1- \sum_i p_i$), let $s_i(\bp) = a_i(\bp) -P(\bp)$, and
let $s_0(\bp) = -P(\bp)$.

\begin{proposition}
If $M$ is strictly truthful (i.e., it is strictly optimal to be truthful), then $S(M)$ is strictly proper.
\end{proposition}

\begin{proof}
By definition, the payoff to the agent when using the scoring rule and reporting $\bpt$ 
is
$$\sum_{1\le i\le n} p_i s_i(\bpt)  + \left(1-\sum_{1\le i\le n} p_i\right) s_0(\bpt)$$
This is the same as
$$
\sum_{1\le i\le n} p_i \left(a_i(\bpt)- P(\bpt)\right)  - \left(1-\sum_{1\le i\le n} p_i\right) P(\bpt)
= \sum_i p_i a_i (\bpt) - P(\bpt).
$$
and thus the incentives for the scoring rule are identical
to the incentives for the mechanism.
\end{proof}

\subsection{Translating Scoring Rules to Mechanisms}

Let $S$ be a non-trivial scoring rule. We assume that the $s_i$'s are of bounded absolute value.
We show how to convert this to a mechanism:

Define the constants $C_0$ and $C$  as follows:
$$C_0 = \max_{\bp} |s_i(\bp) - s_0(\bp)|$$
and
$$C= \max_{\bp} \sum_{1\le i\le n} \left(s_i(\bp) - s_0(\bp) +
C_0\right).$$
Since the scoring rule is non-trivial, the payments (scores) are not constant and therefore $C>0$.
Notice that $s_i(\bp) - s_0(\bp) + C_0 \ge 0$
and that
$$ \sum_{i\le i\le n} \left(\frac{s_i(\bp) - s_0(\bp) + C_0}{C}\right)\le 1.$$

The mechanism $M(S)$ is now defined as follows.
\begin{itemize}
\item The mechanism takes as input the values $x_i$ for each of the alternatives
$1\le i \le n$. We assume that $x_i \ge 0$ for all $i$ and that $\sum_{1\le i \le n} x_i \le 1$.
\item Define $a_i(\bx) = \left(s_i(\bx)-s_0(\bx) + C_0 \right)/C$. As observed above,
$\sum_i a_i(\bx) \le 1$ and $a_i(\bx) \ge 0$.
\item Define $P(\bx) = -\left( s_0 (\bx)+ (1-\sum_i x_i)C_0\right)/C$
\end{itemize}

\begin{proposition}
If $S$ is strictly proper, then $M(S)$ is strictly truthful.
\end{proposition}

\begin{proof}
The utility of a player playing this mechanism and reporting $\bx$
is $$\left(\sum_i x_i a_i (\bx)\right) - P(\bx).$$
This is the same as
$$\sum_i x_i \frac{\left(s_i(\bx) - s_0 (\bx) + C_0\right)}{C} +
\frac{\left(s_0 (\bx)+ \left(1-\sum_i x_i\right)C_0\right)}{C},$$
which is

$$\sum_i x_i \frac{\left(s_i(\bx)  + C_0\right)}{C} +
\left(1-\sum_i x_i\right)\frac{\left(s_0(\bx) + C_0\right)}{C}.$$

If $S$ is strictly proper then, it is strictly proper under the translation by $C_0$ and scaling by $C$.
Thus the utility of the player in the mechanism is strictly maximized by reporting truthfully.
\end{proof}

\subsection{Strong truthfulness}
Let us define $m(x)$-strongly proper scoring rules as follows:
\begin{definition}
A scoring rule with scores $s_i(\bp)$ is  $m(\bp)$-strongly proper
if for every $\bp$, $\bpt$
\[
u(\bp,\bp)-u(\bp,\bpt)\geq \frac{1}{2} m(\bpt)\, ||\bp - \bpt||^2,
\]
where
$$u(\bp, \bpt) = \sum_i p_i s_i(\bpt).$$

\end{definition}

\begin{theorem}

\begin{itemize}
\item Let $M$ be an $m(\bx)$-strongly truthful mechanism. Then $S(M)$ is
an $m(x)$-strongly proper scoring rule.
\item Let $S$ be an $m(\bp)$-strongly proper mechanism. Then $M(S)$
is an $m(x)/C$ strongly truthful mechanism.
\end{itemize}
\end{theorem}

\begin{proof}
The first part is immediate from the fact that utilities are precisely preserved
under the transformation from mechanisms to scoring rules.
For the second part, suppose that for scoring rule $S$
\[
u^S(\bp,\bp)-u^S(\bp,\bpt)\geq \frac{1}{2} m(\bpt)\, ||\bp - \bpt||^2,
\]
Then by the construction above
\[
u^{M(S)}(\bx,\bx)-u^{M(S)}(\bx,\bxt) = \frac{u^S(\bx,\bx)-u^S(\bx,\bxt)}{C} \ge
\frac{1}{2} \frac{ m(\bxt)}{C}\, ||\bx - \bxt||^2.
\]
\end{proof}

\subsection{Application to some standard scoring rules}

\begin{itemize}
\item Logarithmic scoring rule: the translation doesn't work because $C_0$ is unbounded.
\item Quadratic scoring rule: $s_i(\bp) = 1 + 2p_i - ||\bp||^2$.
Then $s_1(p) - s_0(p) = 2(p-(1-p))= 4p - 2 $ which is between -2 and 2. Thus $C_0 = 2$
and $C= 4$.  This translates into $a(x)= x$,  which has $m(x)= 1$.
\item Spherical scoring rule: $s_i(\bp) = p_i/||\bp||$.
Then $s_1(p) - s_0 (p) = (2p-1)/\sqrt{p^2 + (1-p)^2}$ which is between -1 and 1. Thus $C_0=1$
and $C= 2$. This translates into
$$a(x) = \frac{1}{2} + \frac{2x-1}{2\sqrt{x^2 + (1-x)^2}}. $$
%which Amos has a picture of.
The $m(x)$ value is half that of the spherical rule.
\end{itemize}

\end{document}